\documentclass[12pt]{article}
\usepackage{amsfonts}
\newcommand{\Section}[1]%
{\section{#1}\setcounter{equation}{0}%
\setcounter{theorem}{0}}
\newtheorem{theorem}{Theorem}
\newtheorem{lemma}[theorem]{Lemma}

\newtheorem{assumption}[theorem]{Assumption}

\def\re{\mathbb{R}}
\def\co{\mathbb{C}}

\def\ze{\mathbb{Z}}

\newenvironment{proof}[1]%
{\par\noindent{\em #1:\ }}%
{~\rule{2mm}{2mm}\par\bigskip}

\begin{document}
\newpage\thispagestyle{empty}
{\topskip 2cm
\begin{center}
{\Large\bf Power-Law Decay Exponents of Nambu-Goldstone Transverse Correlations\\} 
\bigskip\bigskip
{\Large Tohru Koma\footnote{\small \it Department of Physics, Gakushuin University, Mejiro, Toshima-ku, Tokyo 171-8588, JAPAN,
{\small\tt e-mail: tohru.koma@gakushuin.ac.jp}}
\\}
\end{center}
\vfil
\noindent
{\bf Abstract:} We study a quantum antiferromagnetic Heisenberg model on a hypercubic lattice 
in three or higher dimensions $d\ge 3$. When a phase transition occurs with the continuous symmetry breaking, 
the nonvanishing spontaneous magnetization which is obtained by applying the infinitesimally weak 
symmetry breaking field is equal to the maximum spontaneous magnetization at zero or non-zero low temperatures. 
In addition, the transverse correlation in the infinite-volume limit exhibits a Nambu-Goldstone-type slow decay. 
In this paper, we assume that the transverse correlation decays by power law with distance. 
Under this assumption, we prove that the power is equal to $2-d$ at non-zero low temperatures, 
while it is equal to $1-d$ at zero temperature. 
The method is applied also to a quantum XY model and a classical Heisenberg model at non-zero low temperatures 
in three or higher dimensions. The resulting power is given by the same $2-d$ at non-zero low temperatures. 
\par
\noindent
\bigskip
\hrule
\bigskip\bigskip\bigskip
\vfil}

\Section{Introduction}

In general, when phase transitions occur with continuous symmetry breaking in many-body systems, 
there appear Nambu-Goldstone modes \cite{Nambu,NJL,Goldstone,GSW} which are reflected in a slow decay of the transverse correlations. 
In order to elucidate the universal nature of Nambu-Goldstone modes,  
we study classical and quantum lattice spin systems which exhibit phase transitions at zero or 
low non-zero temperatures with continuous symmetry braking \cite{FSS,DLS,KLS,KLS2,KuboKishi}. 
In the previous paper \cite{Koma}, we studied a quantum antiferromagnetic Heisenberg model as an example, 
and proved that, 
when the spontaneous magnetization is non-vanishing at zero or non-zero low temperatures,  
the transverse correlations in the infinite-volume limit exhibit a Nambu-Goldstone-type slow decay with distance. 
However, this does not necessarily imply a power-law decay of the transverse correlations with large distance. 
In this paper, we make a natural assumption that the transverse correlations show a power-law decay with large distance. 
Under this assumption, we can determine the power of the decay of the transverse correlations. 
More precisely, for the antiferromagnetic Heisenberg model in three or higher dimensions $d\ge 3$, 
the power is given by $2-d$ at non-zero temperatures, while it is given by $1-d$ at zero temperature. 
Our approach is also applied to a quantum XY model and a classical Heisenberg model at non-zero temperatures 
in three or higher dimensions. The resulting power is given by the same $2-d$.  

\Section{Models and results}

In this and the following two sections, we will treat only a quantum Heisenberg antiferromagnet   
which has reflection positivity \cite{DLS,KLS}. 
The treatment of a quantum XY model and a classical Heisenberg model both of which have 
reflection positivity \cite{FSS,DLS} is given in Sec.~\ref{Sec:other}
because our approach to these two models is the same as that to the quantum Heisenberg antiferromagnet. 

Let $\Gamma$ be a finite subset of the $d$-dimensional hypercubic lattice $\ze^d$, i.e., $\Gamma\subset\mathbb{Z}^d$,  
with $d\ge 3$. For each site $x=(x^{(1)},x^{(2)},\ldots,x^{(d)})\in\Gamma$, 
we associate three component quantum spin operator ${\bf S}_x=(S_x^{(1)},S_x^{(2)},S_x^{(3)})$ 
with magnitude of spin, $S=1/2,1,3/2,2,\ldots$. More precisely, the spin operators, $S_x^{(1)}, S_x^{(2)}, S_x^{(3)}$, 
are $(2S+1)\times(2S+1)$ matrices at the site $x$. They satisfy the commutation relations, 
$$
[S_x^{(1)},S_x^{(2)}]=iS_x^{(3)}, \quad [S_x^{(2)},S_x^{(3)}]=iS_x^{(1)}, 
\quad \mbox{and} \quad [S_x^{(3)},S_x^{(1)}]=iS_x^{(2)},
$$
and $(S_x^{(1)})^2+(S_x^{(2)})^2+(S_x^{(3)})^2=S(S+1)$ for $x\in\Gamma$. 
For the finite lattice $\Gamma$, the whole Hilbert space is given by 
$$
\mathfrak{H}_\Gamma=\bigotimes_{x\in\Gamma} \co^{2S+1}.
$$
More generally, the algebra of observables on $\mathfrak{H}_\Gamma$ is given by 
$$
\mathfrak{A}_\Gamma:=\bigotimes_{x\in\Gamma}M_{2S+1}(\co),
$$
where $M_{2S+1}(\co)$ is the algebra of $(2S+1)\times(2S+1)$ complex matrices. 
When two finite lattices, $\Gamma_1$ and $\Gamma_2$, satisfy $\Gamma_1\subset\Gamma_2$, 
the algebra $\mathfrak{A}_{\Gamma_1}$ is embedded in $\mathfrak{A}_{\Gamma_2}$ by 
the tensor product $\mathfrak{A}_{\Gamma_1}\otimes I_{\Gamma_2\backslash\Gamma_1}\subset 
\mathfrak{A}_{\Gamma_2}$ with the identity $I_{\Gamma_2\backslash\Gamma_1}$. 
The local algebra is given by 
$$
\mathfrak{A}_{\rm loc}=\bigcup_{\Gamma\subset\ze^d:|\Gamma|<\infty}\mathfrak{A}_{\Gamma},
$$
where $|\Gamma|$ is the number of the sites in the finite lattice $\Gamma$. 
The quasi-local algebra is defined by the completion of the local algebra $\mathfrak{A}_{\rm loc}$ 
in the sense of the operator-norm topology. 

Consider a $d$-dimensional finite hypercubic lattice, 
\begin{equation}
\label{Lambda}
\Lambda:=\{-L+1,-L+2,\ldots,-1,0,1,\ldots,L-1,L\}^d\subset\mathbb{Z}^d,
\end{equation}
with a large positive integer $L$ and $d\ge 3$. 
The Hamiltonian $H^{(\Lambda)}(B)$ of the Heisenberg antiferromagnet on the lattice $\Lambda$ is given by 
\begin{equation}
\label{H(B)}
H^{(\Lambda)}(B)=H_{0}^{(\Lambda)}-BO^{(\Lambda)},
\end{equation}
where the first term in the right-hand side is the Hamiltonian of the nearest neighbor spin-spin 
antiferromagnetic interactions, 
\begin{equation}
\label{H0}
H_{0}^{(\Lambda)}:=\sum_{\{x,y\}\subset\Lambda:\;|x-y|=1}{\bf S}_x\cdot{\bf S}_y,
\end{equation}
and the second term is the potential due to the external magnetic field $B\in\re$ with the order parameter, 
$$
O^{(\Lambda)}:=\sum_{x\in\Lambda}(-1)^{x^{(1)}+x^{(2)}+\cdots+x^{(d)}}S_x^{(3)}.
$$
Here, we impose the periodic boundary condition.  


The thermal expectation value at the inverse temperature $\beta$ is given by 
\begin{equation}
\label{thermalstate}
\langle\cdots\rangle_{B,\beta}^{(\Lambda)}:=\frac{1}{Z_{B,\beta}^{(\Lambda)}}
{\rm Tr}\;(\cdots)e^{-\beta H^{(\Lambda)}(B)},
\end{equation}
where $Z_{B,\beta}^{(\Lambda)}$ is the partition function. The infinite-volume thermal equilibrium state 
is given by 
\begin{equation}
\label{rhoB}
\rho_{B,\beta}(\cdots)={\rm weak}^\ast\mbox{-}\lim_{\Lambda\nearrow\ze^d}\langle\cdots\rangle_{B,\beta}^{(\Lambda)}.
\end{equation}
Here, if necessary, we take a suitable sequence of the finite lattices $\Lambda$ 
in the weak$^\ast$ limit so that the expectation value converges to a linear functional \cite{Koma}. Similarly, we write 
\begin{equation}
\label{def:rho0}
\rho_{0,\beta}(\cdots):={\rm weak}^\ast\mbox{-}\lim_{B\searrow 0}\rho_{B,\beta}(\cdots).
\end{equation}
Then, the spontaneous magnetization $m_{{\rm s},\beta}$ is given by 
\begin{equation}
\label{msbeta}
m_{{\rm s},\beta}:=\lim_{\Lambda\nearrow\ze^d}\frac{1}{|\Lambda|}\rho_{0,\beta}(O^{(\Lambda)})
\end{equation}
with the hypercubic lattice $\Lambda$ of (\ref{Lambda}) with the even side length $2L$. 
Because of the translational invariance with period 2, the right-hand side dose not depend on 
the size of the lattice $\Lambda$.

The ground-state expectation value is given by 
\begin{equation}
\omega_B^{(\Lambda)}(\cdots):=\lim_{\beta\rightarrow\infty}\langle\cdots\rangle_{B,\beta}^{(\Lambda)}.
\end{equation}
We also write 
\begin{equation}
\omega_0(\cdots):={\rm weak}^\ast\mbox{-}\lim_{B\searrow 0}{\rm weak}^\ast\mbox{-}\lim_{\Lambda\nearrow\ze^d}
\omega_B^{(\Lambda)}(\cdots).
\end{equation}
The corresponding spontaneous magnetization $m_{\rm s}$ at zero temperature is given by 
\begin{equation}
\label{ms}
m_{\rm s}:=\lim_{\Lambda\nearrow\ze^d}\frac{1}{|\Lambda|}\omega_0(O^{(\Lambda)}). 
\end{equation}

In the previous paper \cite{Koma}, we proved that the Nambu-Goldstone transverse correlations show a slow decay with large distance. 
In this paper, we require slightly stronger assumptions at zero and non-zero temperatures as follows:

\begin{assumption}
\label{Assumpnonzero}
We assume that for non-zero temperatures $\beta^{-1}$ satisfying $m_{{\rm s},\beta}>0$, the corresponding transverse correlation 
decays by power law, i.e.,  
\begin{equation}
\label{CorDecayAssumpnonzero}
\left|\rho_{0,\beta}(S_x^{(1)}S_y^{(1)})\right|\sim \frac{K(\beta)}{|x-y|^{d-2+\eta}}
\end{equation}
for a large $|x-y|$ with an exponent $\eta$, where $K(\beta)$ is a positive function of $\beta$. 
\end{assumption}

\noindent
{\it Remark:} Theorem~2.5 in the previous paper \cite{Koma} rules out the possibility of a rapid decay 
$o(|x-y|^{-(d-2)})$ for the transverse correlation $\rho_{0,\beta}(S_x^{(1)}S_y^{(1)})$ at non-zero temperatures, 
where $o(\varepsilon)$ denotes a quantity $q(\varepsilon)$ such that $q(\varepsilon)/\varepsilon$ is vanishing 
in the limit $\varepsilon\searrow 0$. 
Therefore, the exponent $\eta$ must satisfy $2-d<\eta\le 0$.  

\begin{assumption}
\label{Assumpzero}
We assume that $m_{\rm s}>0$, and that the corresponding transverse correlation decays by power law at zero temperature, i.e.,  
\begin{equation}
\left|\omega_0(S_x^{(1)}S_y^{(1)})\right|\sim \frac{{\rm Const.}}{|x-y|^{d-2+\eta'}}
\end{equation}
for a large $|x-y|$ with an exponent $\eta'$. 
\end{assumption}

\noindent
{\it Remark:} Similarly, Theorem~2.3 in the previous paper \cite{Koma} rules out the possibility of a rapid decay 
$o(|x-y|^{-(d-1)})$ for the transverse correlation $\omega_0(S_x^{(1)}S_y^{(1)})$ at zero temperature. 
This implies that the exponent $\eta'$ must satisfy $2-d<\eta'\le 1$.  
\medskip

Our results are as follows: 

\begin{theorem}
\label{theorem:nonzero}
Under Assumption~\ref{Assumpnonzero}, the following is valid: 
If the spontaneous magnetization $m_{{\rm s},\beta}$ is nonvanishing at non-zero temperatures $\beta^{-1}$ 
in dimensions $d\ge 3$, then the corresponding transverse correlation decays by power law,  
\begin{equation}
\rho_{0,\beta}(S_x^{(1)}S_y^{(1)})\sim (-1)^{x^{(1)}+\cdots +x^{(d)}}(-1)^{y^{(1)}+\cdots +y^{(d)}}
\frac{K(\beta)}{|x-y|^{d-2}},
\end{equation}
for a large $|x-y|$, where $K(\beta)$ is a positive function of $\beta$. Namely, the exponent is given by $\eta=0$.  
\end{theorem}

\begin{theorem}
\label{theorem:zero}
Under Assumption~\ref{Assumpzero}, the following is valid: 
If the spontaneous magnetization $m_{\rm s}$ is nonvanishing at zero temperature in dimensions $d\ge 3$, 
then the corresponding transverse correlation decays by power law,   
\begin{equation}
\omega_0(S_x^{(1)}S_y^{(1)})\sim (-1)^{x^{(1)}+\cdots +x^{(d)}}(-1)^{y^{(1)}+\cdots +y^{(d)}}
\frac{K_0}{|x-y|^{d-1}},
\end{equation}
for a large $|x-y|$, where $K_0$ is a positive constant. Namely, the exponent is given by $\eta'=1$. 
\end{theorem} 

\Section{Non-zero temperatures: Proof of Theorem~\ref{theorem:nonzero}}
\label{Sec:nonzero}

Consider first the case of non-zero temperatures, and we will give a proof of Theorem~\ref{theorem:nonzero} 
in this section. 
To begin with, we recall the previous result of \cite{Koma} about the transverse correlation.  
{From} Theorem~2.5 of \cite{Koma}, the rapid decay $o(|x-y|^{-(d-2)})$ for the transverse correlation is excluded.  
This implies that the exponent $\eta$ of (\ref{CorDecayAssumpnonzero}) must satisfy $2-d<\eta\le 0$. 
Therefore, in order to prove Theorem~\ref{theorem:nonzero}, it is sufficient to show $\eta\ge 0$ except for 
the prefactor. 

\begin{lemma}
The following bound is valid: 
\begin{equation}
\label{S1-cor}
(-1)^{x^{(1)}+\cdots +x^{(d)}}(-1)^{y^{(1)}+\cdots +y^{(d)}}\langle S_x^{(1)}S_y^{(1)}\rangle_{B,\beta}^{(\Lambda)}\ge 0
\end{equation}
for any $x,y\in\Lambda$. 
\end{lemma}

\begin{proof}{Proof} We first introduce a unitary transformation $U$ which causes $\pi$ rotation of the spins ${\bf S}_x$ 
about the 2 axis for the sites $x$ with odd $(x^{(1)}+\cdots + x^{(d)})$. As a result, the Hamiltonian is transformed into \cite{DLS} 
\begin{equation}
\label{tildeH}
\tilde{H}^{(\Lambda)}(B):=U^\ast H^{(\Lambda)}(B)U=\sum_{|x-y|=1}(-S_x^{(1)}S_y^{(1)}+S_x^{(2)}S_y^{(2)}
-S_x^{(3)}S_y^{(3)})-B\sum_x S_x^{(3)}. 
\end{equation}
The left-hand side of (\ref{S1-cor}) can be written 
\begin{equation}
\label{S1tildeCor}
(-1)^{x^{(1)}+\cdots +x^{(d)}}(-1)^{y^{(1)}+\cdots +y^{(d)}}\langle S_x^{(1)}S_y^{(1)}\rangle_{B,\beta}^{(\Lambda)}
=\langle\!\langle S_x^{(1)}S_y^{(1)}\rangle\!\rangle_{B,\beta}^{(\Lambda)},
\end{equation}
where we have written  
$$
\langle\!\langle \cdots\rangle\!\rangle_{B,\beta}^{(\Lambda)}
:=\frac{1}{\tilde{Z}_{B,\beta}^{(\Lambda)}}{\rm Tr}(\cdots)\; e^{-\beta \tilde{H}^{(\Lambda)}(B)}
$$
with $\tilde{Z}_{B,\beta}^{(\Lambda)}:={\rm Tr}\;e^{-\beta \tilde{H}^{(\Lambda)}(B)}$. 
Therefore, it is sufficient to show that the right-hand side of (\ref{S1tildeCor}) is nonnegative. 

The Hamiltonian $\tilde{H}^{(\Lambda)}(B)$ of (\ref{tildeH}) can be decomposed into two parts, 
$$
\tilde{H}^{(\Lambda)}(B)=\tilde{H}_{\rm XY}^{(\Lambda)}+\tilde{H}_{\rm Ising}^{(\Lambda)}(B),
$$
with 
$$
\tilde{H}_{\rm XY}^{(\Lambda)}:=-\sum_{|x-y|=1}(S_x^{(1)}S_y^{(1)}-S_x^{(2)}S_y^{(2)})
$$
and 
$$
\tilde{H}_{\rm Ising}^{(\Lambda)}(B):=-\sum_{|x-y|=1}S_x^{(3)}S_y^{(3)}-B\sum_x S_x^{(3)}.
$$
As usual, we write $S_x^{(\pm)}:=S_x^{(1)}\pm iS_x^{(2)}$. Then, one has 
\begin{equation}
S_x^{(1)}S_y^{(1)}-S_x^{(2)}S_y^{(2)}=\frac{1}{2}\left[S_x^{(+)}S_y^{(+)}+S_x^{(-)}S_y^{(-)}\right]. 
\end{equation}
All the matrix elements of this right-hand side are nonnegative in the usual real, orthonormal basis 
which diagonalizes $S_x^{(3)}$. Therefore, all the matrix elements of $-\tilde{H}_{\rm XY}^{(\Lambda)}$ 
are nonnegative. On the other hand, one has 
$$
e^{-\beta\tilde{H}^{(\Lambda)}(B)}=\lim_{M\rightarrow\infty}\left[e^{-\beta\tilde{H}_{\rm XY}^{(\Lambda)}/M}
e^{-\beta\tilde{H}_{\rm Ising}^{(\Lambda)}(B)/M}\right]^M
$$
{from} the above decomposition of the Hamiltonian $\tilde{H}^{(\Lambda)}(B)$. Therefore, from the above observation, 
this left-hand side $e^{-\beta\tilde{H}^{(\Lambda)}(B)}$ has only nonnegative matrix elements. 
Since all the matrix elements of $S_x^{(1)}$ are also nonnegative in the same basis, these observations imply 
the right-hand side of (\ref{S1tildeCor}) is nonnegative.   
\end{proof}

By using the assumption (\ref{CorDecayAssumpnonzero}), the positivity (\ref{S1-cor}) and the relation (\ref{S1tildeCor}), we have 
\begin{equation}
\label{lboundCor}
(-1)^{x^{(1)}+\cdots +x^{(d)}}(-1)^{y^{(1)}+\cdots +y^{(d)}}\rho_{0,\beta}(S_x^{(1)}S_y^{(1)})
=\tilde{\rho}_{0,\beta}(S_x^{(1)}S_y^{(1)})\ge \frac{{\cal C}_1}{|x-y|^{d-2+\eta}}
\end{equation}
for any $x,y$ satisfying $|x-y|\ge r_0$ with a large positive constant $r_0$, where 
${\cal C}_1$ is a positive constant which may depend on $\beta$, and we have written 
\begin{equation}
\label{tilderho}
\tilde{\rho}_{0,\beta}(\cdots):={\rm weak}^\ast\mbox{-}\lim_{B\searrow 0}{\rm weak}^\ast\mbox{-}\lim_{\Lambda\nearrow \ze^d}
\langle\!\langle\cdots\rangle\!\rangle_{B,\beta}^{(\Lambda)}.
\end{equation}

For an operator $A$, we introduce three quantities as \cite{DLS}  
\begin{equation}
\label{g}
g_{B,\beta}^{(\Lambda)}(A):=\frac{1}{2}\left[\langle\!\langle A A^\ast \rangle\!\rangle_{B,\beta}^{(\Lambda)}+
\langle\!\langle A^\ast A \rangle\!\rangle_{B,\beta}^{(\Lambda)}\right],
\end{equation}
$$
b_{B,\beta}^{(\Lambda)}(A):=(A^\ast,A)_{B,\beta}^{(\Lambda)}:=\frac{1}{\tilde{Z}_{B,\beta}^{(\Lambda)}}\int_0^1 ds\; 
{\rm Tr}\left[A^\ast\; e^{-s \beta \tilde{H}^{(\Lambda)}(B)}\; A\; 
e^{-(1-s)\beta \tilde{H}^{(\Lambda)}(B)}\right]
$$
and 
\begin{equation}
\label{c}
c_{B,\beta}^{(\Lambda)}(A):=\langle\!\langle[A^\ast,[\tilde{H}^{(\Lambda)}(B),A]]\rangle\!\rangle_{B,\beta}^{(\Lambda)}.
\end{equation}
The method of the reflection positivity \cite{DLS} is applicable to the present system with the Hamiltonian $\tilde{H}^{(\Lambda)}(B)$. 
As a result, the following bound \cite{DLS} is valid:  
\begin{equation}
\label{DLSIRbound}
\Bigl(\sigma\Bigl(\sum_{i=1}^d \overline{\partial_i h_i}\Bigr),
\sigma\Bigl(\sum_{i=1}^d \partial_i h_i\Bigr)\Bigr)_{B,\beta}^{(\Lambda)}
\le \beta^{-1}\sum_{i=1}^d \sum_{x\in \Lambda}|h_i(x)|^2
\end{equation}
where $h_i$ are complex-valued functions on $\Lambda$, $\partial_i h_j(x):=h_j(x+e_i)-h_j(x)$ with 
the unit lattice vector $e_i$ whose $k$-th component is given by $e_i^{(k)}=\delta_{i,k}$, and 
$$
\sigma(f)=\sum_{x\in \Lambda}f(x)S_x^{(1)} 
$$
for a function $f$ on $\Lambda$. Here, $\overline{\cdots}$ denotes complex conjugate. 
In addition to this, the function $b_{B,\beta}^{(\Lambda)}(A)$ satisfies \cite{FB} 
\begin{equation}
\label{bpboundgpcp}
b_{B,\beta}^{(\Lambda)}(A)\ge \frac{4[g_{B,\beta}^{(\Lambda)}(A)]^2}{4g_{B,\beta}^{(\Lambda)}(A)+\beta c_{B,\beta}^{(\Lambda)}(A)}, 
\end{equation}
where we have used the inequalities (34) and (A10) in \cite{DLS}, and 
$$
t^{-1}(1-e^{-t})\ge (1+t)^{-1} \quad \mbox{for \ } t>0.
$$
Using the inequality (\ref{bpboundgpcp}), the function $g_{B,\beta}^{(\Lambda)}(A)$ is estimated as  
\begin{equation}
\label{gbound}
g_{B,\beta}^{(\Lambda)}(A)\le \frac{1}{2}\left\{b_{B,\beta}^{(\Lambda)}(A)
+\sqrt{[b_{B,\beta}^{(\Lambda)}(A)]^2+\beta b_{B,\beta}^{(\Lambda)}(A) c_{B,\beta}^{(\Lambda)}(A)}\right\}. 
\end{equation}

First, we recall the following well known facts \cite{FSS}: Note that 
\begin{eqnarray*}
\langle \partial_i\varphi,\psi\rangle &=&\sum_{x\in \ze^d}[\overline{\varphi(x+e_i)}-\overline{\varphi(x)}]\psi(x)\\
&=&\sum_x \overline{\varphi(x+e_i)}\psi(x) - \sum_x \overline{\varphi(x)}\psi(x)\\
&=&\sum_x \overline{\varphi(x)}\psi(x-e_i)- \sum_x \overline{\varphi(x)}\psi(x)\\
&=&-\sum_x \overline{\varphi(x)}[\psi(x)-\psi(x-e_i)]
\end{eqnarray*}
for functions, $\varphi$ and $\psi$, on $\ze^d$ with a compact support. Therefore, the adjoint $\partial_i^\ast$ of 
$\partial_i$ is given by 
$$
\partial_i^\ast \psi(x)=-[\psi(x)-\psi(x-e_i)].
$$
The corresponding Laplacian $\Delta$ is given by 
$$
\Delta \varphi(x)=\sum_{i=1}^d [\varphi(x+e_i)+\varphi(x-e_i)-2\varphi(x)].
$$
Actually, one has 
$$
-\Delta=\sum_{i=1}^d \partial_i^\ast \partial_i =\sum_{i=1}^d \partial_i\partial_i^\ast. 
$$
The inverse of $\Delta$ is given by 
$$
\Delta^{-1}(x,y)=\frac{1}{(2\pi)^d}\int_{-\pi}^\pi dk_1\cdots \int_{-\pi}^\pi dk_d 
\frac{e^{ik(x-y)}}{E_k},
$$
where 
$$
E_k:= 2\sum_{i=1}^d[\cos k_i-1]. 
$$
Clearly, $\Delta^{-1}$ is well defined for $d\ge 3$. 

Following \cite{FSS}, we choose $h_i=-\Delta^{-1}\partial_i^\ast \chi_\Omega$ 
with the characteristic function $\chi_\Omega$ of a cube $\Omega$ with the sidelength $R$. Then, one has 
\begin{equation}
\sum_{i=1}^d \partial_ih_i=-\sum_{i=1}^d \partial_i\Delta^{-1}\partial_i^\ast \chi_\Omega
=\chi_\Omega.
\end{equation}
This implies 
\begin{equation}
\label{sigmaparhi}
\sigma\left(\sum_{i=1}^d \partial_i h_i\right)=\sigma(\chi_\Omega)
=\sum_{x\in\ze^d} \chi_\Omega(x)S_x^{(1)}=\sum_{x\in\Omega}S_x^{(1)}.
\end{equation}
Further,
\begin{eqnarray}
\sum_{i=1}^d \sum_{x\in\ze^d} |h_i(x)|^2
&=&\sum_{i=1}^d \langle \Delta^{-1}\partial_i^\ast\chi_\Omega, \Delta^{-1}\partial_i^\ast\chi_\Omega\rangle \nonumber\\
&=&\sum_{i=1}^d \langle \chi_\Omega,\partial_i\Delta^{-2}\partial_i^\ast\chi_\Omega\rangle\nonumber\\
&=& \langle\chi_\Omega,(-\Delta^{-1})\chi_\Omega\rangle. 
\end{eqnarray}
Since $(-\Delta^{-1})(x,y)\sim |x-y|^{2-d}$ for a large $|x-y|$, 
\begin{equation}
\label{Laplaceinvbound}
\langle\chi_\Omega,(-\Delta^{-1})\chi_\Omega\rangle\le {\cal C}_2 R^{d+2} 
\end{equation}
for a large sidelength $R$ of the box $\Omega$, where ${\cal C}_2$ is a positive constant. 

Next, in order to handle the finite lattice $\Lambda$ with the periodic boundary condition, 
we set 
\begin{equation}
\label{hLambda}
h^{(\Lambda)}:=\chi_\Omega-\frac{|\Omega|}{|\Lambda|}\chi_\Lambda
\end{equation}
for $\Omega\subset\Lambda$. Clearly, one has $\sum_{x\in\Lambda}h^{(\Lambda)}(x)=0$. Therefore, one can define 
\begin{equation}
h_i^{(\Lambda)}:=-(\Delta^{(\Lambda)})^{-1}\partial_i^\ast h^{(\Lambda)},
\end{equation}
where $\Delta^{(\Lambda)}$ is the Laplacian for the finite lattice $\Lambda$ with the periodic boundary condition. 
Further, one obtains 
$$
\sum_i \partial_i h_i^{(\Lambda)}=h^{(\Lambda)}
$$
and 
$$
\sum_i \langle h_i^{(\Lambda)},h_i^{(\Lambda)}\rangle = \langle h^{(\Lambda)},[-(\Delta^{(\Lambda)})^{-1}]h^{(\Lambda)}\rangle.
$$
Substituting these into (\ref{DLSIRbound}), we have 
\begin{equation}
\label{DuhamLaplaceBound}
(\sigma(h^{(\Lambda)}),\sigma(h^{(\Lambda)}))_{B,\beta}^{(\Lambda)}
\le \beta^{-1} \langle h^{(\Lambda)},[-(\Delta^{(\Lambda)})^{-1}]h^{(\Lambda)}\rangle.
\end{equation}
In the right-hand side, 
\begin{equation}
\label{LaplaceLimit}
\langle h^{(\Lambda)},[-(\Delta^{(\Lambda)})^{-1}]h^{(\Lambda)}\rangle\rightarrow 
\langle \chi_\Omega,(-\Delta^{-1})\chi_\Omega\rangle \ \ \ \mbox{as} \ \ \Lambda\nearrow\ze^d.
\end{equation} 
Since $h^{(\Lambda)}$ is given by (\ref{hLambda}), the left-hand side is written 
\begin{eqnarray}
\label{Duhamdocom}
(\sigma(h^{(\Lambda)}),\sigma(h^{(\Lambda)}))_{B,\beta}^{(\Lambda)}
&=&(\sigma(\chi_\Omega),\sigma(\chi_\Omega))_{B,\beta}^{(\Lambda)}
-\frac{|\Omega|}{|\Lambda|}(\sigma(\chi_\Omega),\sigma(\chi_\Lambda))_{B,\beta}^{(\Lambda)}\nonumber \\
&-&\frac{|\Omega|}{|\Lambda|}(\sigma(\chi_\Lambda),\sigma(\chi_\Omega))_{B,\beta}^{(\Lambda)}
+\frac{|\Omega|^2}{|\Lambda|^2}(\sigma(\chi_\Lambda),\sigma(\chi_\Lambda))_{B,\beta}^{(\Lambda)}.
\end{eqnarray}
The first term in the right-hand side is written 
\begin{equation}
\label{Duhamdoco1st}
(\sigma(\chi_\Omega),\sigma(\chi_\Omega))_{B,\beta}^{(\Lambda)}
=(A_\Omega,A_\Omega)_{B,\beta}^{(\Lambda)}=b_{B,\beta}^{(\Lambda)}(A_\Omega),
\end{equation}
where we have written  
$$
A_\Omega:=\sigma(\chi_\Omega)=\sum_{x\in\Lambda}\chi_\Omega(x)S_x^{(1)}=\sum_{x\in\Omega}S_x^{(1)}
$$
for short. In order to estimate the second term in the right-hand side, we note that \cite{DLS} 
\begin{equation}
\label{Duhambound}
|(\sigma(\chi_\Omega),\sigma(\chi_\Lambda))_{B,\beta}^{(\Lambda)}|
\le \sqrt{\langle\!\langle \sigma(\chi_\Omega)^2\rangle\!\rangle_{B,\beta}^{(\Lambda)}}
\sqrt{\langle\!\langle\sigma(\chi_\Lambda)^2\rangle\!\rangle_{B,\beta}^{(\Lambda)}}
=\sqrt{\langle\!\langle A_\Omega^2\rangle\!\rangle_{B,\beta}^{(\Lambda)}}
\sqrt{\langle\!\langle A_\Lambda^2\rangle\!\rangle_{B,\beta}^{(\Lambda)}}. 
\end{equation}
Using the translation invariance and Schwarz inequality, we have 
$$
\frac{1}{|\Lambda|^2}\langle\!\langle A_\Lambda^2\rangle\!\rangle_{B,\beta}^{(\Lambda)}
=\frac{1}{|\Gamma||\Lambda|}\langle\!\langle A_\Gamma A_\Lambda\rangle\!\rangle_{B,\beta}^{(\Lambda)}
\le \frac{1}{|\Gamma||\Lambda|}\sqrt{\langle\!\langle A_\Gamma^2 \rangle\!\rangle_{B,\beta}^{(\Lambda)}}
\sqrt{\langle\!\langle A_\Lambda^2 \rangle\!\rangle_{B,\beta}^{(\Lambda)}}
$$
for $\Gamma\subset\Lambda$. This implies \cite{KomaTasaki} 
\begin{equation}
\frac{1}{|\Lambda|^2}\langle\!\langle A_\Lambda^2\rangle\!\rangle_{B,\beta}^{(\Lambda)}
\le \frac{1}{|\Gamma|^2}\langle\!\langle A_\Gamma^2\rangle\!\rangle_{B,\beta}^{(\Lambda)}.
\end{equation}
Therefore, 
\begin{eqnarray}
\label{LROinequa}
\frac{1}{|\Gamma|^2}\tilde{\rho}_{0,\beta}(A_\Gamma^2)&=&
{\rm weak}^\ast\mbox{-}\lim_{B\searrow 0}{\rm weak}^\ast\mbox{-}\lim_{\Lambda\nearrow\ze^d}
\frac{1}{|\Gamma|^2}\langle\!\langle A_\Gamma^2\rangle\!\rangle_{B,\beta}^{(\Lambda)}\nonumber\\
&\ge& {\rm weak}^\ast\mbox{-}\lim_{B\searrow 0}{\rm weak}^\ast\mbox{-}\lim_{\Lambda\nearrow\ze^d}
\frac{1}{|\Lambda|^2}\langle\!\langle A_\Lambda^2\rangle\!\rangle_{B,\beta}^{(\Lambda)}, 
\end{eqnarray}
where the two weak$^\ast$ limits are taken so that 
$\langle\!\langle\cdots\rangle\!\rangle_{B,\beta}^{(\Lambda)}$ converges to $\tilde{\rho}_{0,\beta}(\cdots)$.   
Further, by the argument in Sec.~7 of \cite{Koma}, we obtain  
$$
\lim_{\Gamma\nearrow\ze^d}\frac{1}{|\Gamma|^2}\tilde{\rho}_{0,\beta}(A_\Gamma^2)=0.
$$
Therefore, we have  
\begin{equation}
\label{NoLROtrns}
{\rm weak}^\ast\mbox{-}\lim_{B\searrow 0}{\rm weak}^\ast\mbox{-}\lim_{\Lambda\nearrow\ze^d}
\frac{1}{|\Lambda|^2}\langle\!\langle A_\Lambda^2\rangle\!\rangle_{B,\beta}^{(\Lambda)}=0. 
\end{equation}
{from} the above inequality (\ref{LROinequa}). Namely, the long-range order of the transverse correlations 
is vanishing. {From} (\ref{Duhambound}) and (\ref{NoLROtrns}), the second term in the right-hand side of  
(\ref{Duhamdocom}) is vanishing in the double weak$^\ast$ limit. Similarly, the third and fourth terms are also 
vanishing. Combining these observations, (\ref{DuhamLaplaceBound}), (\ref{LaplaceLimit}) and (\ref{Duhamdoco1st}), we obtain 
\begin{equation}
{\rm weak}^\ast\mbox{-}\lim_{B\searrow 0}{\rm weak}^\ast\mbox{-}\lim_{\Lambda\nearrow\ze^d}
b_{B,\beta}^{(\Lambda)}(A_\Omega)\le \beta^{-1}\langle \chi_\Omega,(-\Delta^{-1})\chi_\Omega\rangle. 
\end{equation}
The right-hand side is estimated by (\ref{Laplaceinvbound}). As a result, we have 
\begin{equation}
\label{bbound}
{\rm weak}^\ast\mbox{-}\lim_{B\searrow 0}{\rm weak}^\ast\mbox{-}\lim_{\Lambda\nearrow\ze^d}
b_{B,\beta}^{(\Lambda)}(A_\Omega)
\le {\cal C}_2\beta^{-1}R^{d+2}.
\end{equation}
{From} (\ref{c}), one has 
\begin{equation}
{\rm weak}^\ast\mbox{-}\lim_{B\searrow 0}{\rm weak}^\ast\mbox{-}\lim_{\Lambda\nearrow\ze^d}
c_{B,\beta}^{(\Lambda)}(A_\Omega)\le {\cal C}_3 R^d.  
\end{equation}
Further, from (\ref{tilderho}) and (\ref{g}), we have 
\begin{equation}
\tilde{\rho}_{0,\beta}(A_\Omega^2)=
{\rm weak}^\ast\mbox{-}\lim_{B\searrow 0}{\rm weak}^\ast\mbox{-}\lim_{\Lambda\nearrow\ze^d}g_{B,\beta}^{(\Lambda)}(A_\Omega). 
\end{equation}
This left-hand side satisfies 
\begin{equation}
{\cal C}_4R^{d+2-\eta}\le \tilde{\rho}_{0,\beta}(A_\Omega^2)
\end{equation}
{from} the lower bound for the correlation (\ref{lboundCor}). Substituting these into the inequality (\ref{gbound}), 
we obtain 
\begin{equation}
{\cal C}_4R^{d+2-\eta}\le \frac{1}{2}\left[{\cal C}_2\beta^{-1}R^{d+2} 
+\sqrt{{\cal C}_2^2\beta^{-2}R^{2(d+2)}+{\cal C}_2R^{d+2}\cdot{\cal C}_3R^d}\right]
\end{equation}
for a large $R$. This implies $\eta\ge 0$. Since $\eta\le 0$ as mentioned above \cite{Koma}, we obtain $\eta=0$. 

\Section{Zero temperature: Proof of Theorem~\ref{theorem:zero}}

Next consider the transverse correlation at zero temperature. To begin with, we recall that 
the ground-state expectation value is given by 
\begin{equation}
\omega_B^{(\Lambda)}(\cdots)=
\lim_{\beta\rightarrow\infty}\langle\cdots\rangle_{B,\beta}^{(\Lambda)}=\lim_{\beta\rightarrow\infty}\frac{1}{Z_{B,\beta}^{(\Lambda)}}
{\rm Tr}\;(\cdots)e^{-\beta H^{(\Lambda)}(B)}. 
\end{equation}
{From} the bound (\ref{DLSIRbound}), one has 
$$
\lim_{\beta\rightarrow \infty}b_{B,\beta}^{(\Lambda)}(A_\Omega)=0.
$$
Therefore, the key bound (\ref{gbound}) is replaced by \cite{JNFP}
\begin{equation}
\lim_{\beta\rightarrow\infty}g_{B,\beta}^{(\Lambda)}(A_\Omega)
\le \lim_{\beta\rightarrow\infty}\frac{1}{2}\sqrt{\beta b_{B,\beta}^{(\Lambda)}(A_\Omega) c_{B,\beta}^{(\Lambda)}(A_\Omega)}. 
\end{equation}
The same argument as in the case of non-zero temperatures yields 
\begin{equation}
{\cal C}_4R^{d+2-\eta'}\le \frac{1}{2}\sqrt{{\cal C}_2{\cal C}_3}R^{d+1}
\end{equation}
for a large $R$. Here, the positive constant ${\cal C}_4$ may be different from that in 
the case of non-zero temperatures. 
This inequality implies $\eta'\ge 1$. However, $\eta'\le 1$ from the result of \cite{Koma}. 
Thus, we obtain $\eta'=1$ at zero temperature. 

\Section{Other models}
\label{Sec:other}

In this section, we will treat the quantum XY model and the classical Heisenberg model in three or higher dimensions 
$d\ge 3$. For the transverse correlation at non-zero temperatures, we can obtain the same exponent $\eta=0$ as 
that of the quantum antiferromagnetic Heisenberg model under the power-law decay assumption for the transverse correlations.   
Basically, we use Bogoliubov inequalities, the bounds derived from the reflection positivity, and  
the Griffiths-type positivity of the correlations. 

\subsection{Quantum XY model}

The Hamiltonian of the quantum XY model is given by 
$$
H^{(\Lambda)}(B)= -\sum_{|x-y|=1}[S_x^{(1)}S_y^{(1)}+S_x^{(3)}S_y^{(3)}]-B\sum_x S_x^{(3)}. 
$$
This Hamiltonian is different from (\ref{tildeH}) in only the terms $S_x^{(2)}S_y^{(2)}$ about the second component of 
the spins. 
Therefore, both of the reflection positivity \cite{DLS} and the positivity of the transverse correlation hold also 
for the XY model. Clearly, the Bogoliubov inequality holds as well. Thus, the same argument yields 
the exponent $\eta=0$ in the case of non-zero temperatures.  

For the case of zero temperature, we can obtain the lower bound $\eta'\ge 1$ by taking the limit $\beta\rightarrow\infty$ 
in the same way. Unfortunately, we have not been able to find a useful analogue of 
the Kennedy-Lieb-Shastry inequality \cite{KLS} for the ground state of the XY model. 
Therefore, the upper bound of $\eta'$ is missing at zero temperature.

\subsection{Classical Heisenberg model}

The Hamiltonian of the classical Heisenberg model with $N$ component spins is given by \cite{FSS} 
$$
H^{(\Lambda)}(B)=-\sum_{|x-y|}\mbox{\boldmath $\sigma$}_x\cdot \mbox{\boldmath $\sigma$}_y
-B\sum_x \sigma^{(N)}_x, 
$$
where $\mbox{\boldmath $\sigma$}_x=(\sigma_x^{(1)},\sigma_x^{(2)},\ldots,\sigma_x^{(N)})\in\re^N$ with integer $N\ge 2$. 
The model satisfies the reflection positivity, and the positivity of the transverse correlations holds \cite{FSS}. 
Using these properties, the lower bound $\eta\ge 0$ was obtained for a correlation which decays by power law in \cite{FSS}, 
although the situation\footnote{For the field theoretical case, see \cite{GJ}.} which they considered 
is different from ours. The upper bound $\eta\le 0$ can be obtained by using 
the classical Bogoliubov inequality\footnote{The first classical analogs of Bogoliubov inequalities 
were found by Mermin \cite{Mermin}. For a more sophisticated treatment, see, e.g., Sec.~III.6 of the book \cite{Simon}.}
in the same way as in \cite{Koma}. Thus, $\eta=0$ at non-zero temperatures.


\end{document}